\documentclass[runningheads, envcountsame, a4paper]{article}

\usepackage{amsfonts,amsmath,amssymb,amsthm}
\usepackage{xcolor,hyperref}
\usepackage{enumerate,textcomp}
\usepackage{xspace}
\usepackage{ifthen}
\usepackage{pgf,tikz}
\usetikzlibrary{arrows,automata,shapes}
\tikzstyle{every state}=[minimum size=12pt,inner sep=0pt]
\usetikzlibrary{snakes}
\usetikzlibrary{calc}

\newtheorem{theorem}{Theorem}
\newtheorem{lemma}[theorem]{Lemma}
\newtheorem{proposition}[theorem]{Proposition}

\newtheorem{remark}[theorem]{Remark}

\newcommand{\FR}{\(\mathtt{FR}\)\xspace}
\newcommand{\SK}{\(\mathtt{automgrp}\)\xspace}
\newcommand{\GAP}{\(\mathtt{GAP}\)\xspace}

\newcommand{\mz}{\mathfrak m}

\newcommand{\aut}[1]{{\mathcal #1}}
\newcommand{\dual}[1]{{\mathfrak d}({#1})}
\newcommand{\mot}[1]{{\mathbf {#1}}}
\newcommand{\pres}[1]{\langle{#1}\rangle}
\newcommand{\presm}[1]{\pres{{#1}}_{+}}

\newcommand{\ie}{\emph{i.e.}\xspace}
\newcommand{\wrt}{w.r.t. }

\usepackage{xcolor}

\begin{document}

\title{On level-transitivity and exponential growth\thanks{This work was partially supported by the French \emph{Agence
    Nationale pour la~Recherche}, through the Project {\bf MealyM}
  ANR-JS02-012-01.}}

\author{Ines Klimann\\
Univ Paris Diderot, Sorbonne Paris Cit\'e, IRIF,\\ UMR 8243 CNRS, F-75013 Paris, France\\
{\small \url{klimann@liafa.univ-paris-diderot.fr}}}%
\date{}

\setlength{\parindent}{0pt}

\maketitle

\begin{abstract}
We prove that if the group generated by an invertible and reversible Mealy automaton acts
level-transitively on a regular rooted tree, then the semigroup
generated by the dual automaton has exponential growth, hence giving a
decision procedure of exponential growth for a restricted family of
automaton (semi)groups.
\end{abstract}

The purpose of this note is to link up two classes of
groups and semigroups highly studied for themselves: level-transitive
(semi)groups and (semi)groups of exponential growth, through automaton
(semi)groups.

On the one hand, level-transitive groups (or equivalently spherically transitive
groups, depending on the authors) --- \ie groups acting transitively on every level
of a regular rooted tree --- have received special focus these last years because of
branch groups, which form a particular class of level-transitive groups,
one of the three classes into which the
class of just infinite groups is naturally decomposed
~\cite{Grigorchuk2000,bgs03}.

On the other hand, the study on how (semi)groups grow has been highlighted since Milnor's
question on the existence of groups of
intermediate growth in 1968~\cite{milnor} and the very first example of such a
group given by Grigorchuk~\cite{grigorch:degrees}.

 In this note, we
prove that no semigroup of polynomial or intermediate growth can be generated by an
invertible and reversible Mealy automaton whose dual generates a level-transitive
group. Even if the problem of deciding the level-transitivity of an
automaton group is still open, there exist some families of
Mealy automata for which the level-transitivity of an element
in the generated semigroup is decidable~\cite{Steinb15}.

\section{Basic notions}\label{sec-basic}

\subsection{Semigroups and groups generated by Mealy automata}
We first recall the formal definition of an automaton. A {\em (finite, deterministic, and complete) automaton} is a
triple
\(
\bigl( Q,\Sigma,\delta = (\delta_i\colon Q\rightarrow Q )_{i\in \Sigma} \bigr)
\),
where the \emph{state set}~$Q$
and the \emph{alphabet}~$\Sigma$ are non-empty finite sets, and
the~\(\delta_i\) are functions.

A \emph{Mealy automaton} is a quadruple
\(( Q, \Sigma, \delta
,\rho
)\), 
such that \((Q,\Sigma,\delta)\) and~\((\Sigma,Q,\rho)\) are both
automata.
In other terms, a Mealy automaton is a complete, deterministic,
letter-to-letter transducer with the same input and output
alphabet. Its \emph{size} is the cardinality of its state set.

The graphical representation of a Mealy automaton is
standard, see Figure~\ref{fig-exLaurent}.

\begin{figure}[h]
\centering
\begin{tikzpicture}[->,>=latex,node distance=1.8cm]
\tikzstyle{every state}=[minimum size=12pt,inner sep=0pt]
\node[state] (1) {\(x\)};
\node[state] (4) [right of=1] {\(t\)};
\node[state] (3) [below of=1] {\(z\)};
\node[state] (2) [below of=4] {\(y\)};
\path (1) edge node[left] {\(1|1\)} (3)
      (3) edge node[below]{\(1|0\)} (2)
      (2) edge node[right]{\(1|1\)} (4)
      (4) edge node[above]{\(1|0\)} (1)
      (1) edge [loop left] node{\(0|0\)} (1)
      (2) edge [loop right] node{\(0|0\)} (2)
      (3) edge [loop left] node{\(0|1\)} (3)
      (4) edge [loop right] node{\(0|1\)} (4);
\end{tikzpicture}
\caption{An example of a Mealy automaton which does not generate a
  free semigroup on its state set, but whose
dual generates a level-transitive group.}\label{fig-exLaurent}
\end{figure}
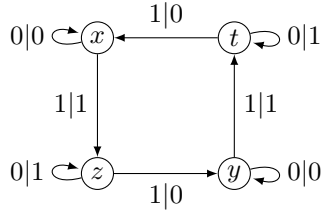

Let~\(\aut{A} = (Q,\Sigma, \delta,\rho)\) be a Mealy automaton.
For each state \(x\in Q\), the map \(\rho_x:\Sigma\to\Sigma\) can be
extended to a map \(\rho_x:\Sigma^*\to\Sigma^*\) recursively defined
by:
\begin{equation*}
\forall i \in \Sigma, \ \forall \mot{s} \in \Sigma^*, \qquad
\rho_x(i\mot{s}) = \rho_x(i)\rho_{\delta_i(x)}(\mot{s}) \:.
\end{equation*}

The image of the empty word is itself.
The mapping~\(\rho_x\) for each $x\in Q$ is length-preserving and prefix-preserving.
We say that~\(\rho_x\) is the function \emph{induced} by \(x\).
For~$\mot{x}=x_1\cdots x_n \in Q^n$ with~$n>0$, set
\(\rho_\mot{x}\colon\Sigma^* \rightarrow \Sigma^*, \rho_\mot{x} = \rho_{x_n}
\circ \cdots \circ \rho_{x_1} \:\).

The semigroup of mappings from~$\Sigma^*$ to~$\Sigma^*$ generated by
$\{\rho_x, x\in Q\}$ is called the \emph{semigroup generated
  by~$\aut{A}$} and is denoted by~$\presm{\aut{A}}$.

\medskip

A Mealy automaton \(\aut{A}=(Q,\Sigma,\delta, \rho)\) is
\emph{invertible\/} if the functions \(\rho_x\) are permutations
of~\(\Sigma\). In this case, the functions induced by the states are
permutations on words of the same length and thus we may consider
the group of mappings from~$\Sigma^*$ to~$\Sigma^*$ generated by
$\{\rho_x, x\in Q\}$: it is called the \emph{group generated
  by~$\aut{A}$} and is denoted by~$\pres{\aut{A}}$.

\medskip

In a Mealy automaton~\(\aut{A}=(Q,\Sigma, \delta, \rho)\), the sets~\(Q\)
and~\(\Sigma\) play dual roles. So we may consider the \emph{dual (Mealy)
automaton} defined by
\(
\dual{\aut{A}} = (\Sigma,Q, \rho, \delta)
\).
We extend to \(\delta\) the former notations on
\(\rho\), in a natural way. Hence \(\delta_i\colon Q^*\rightarrow Q^*,
i\in \Sigma\), are the functions induced by the states
of~$\dual{\aut{A}}$, and for~$\mot{s}=s_1\cdots s_n
\in \Sigma^n$ with~$n>0$, we set~\(\delta_\mot{s}\colon Q^* \rightarrow Q^*,
\ \delta_\mot{s} = \delta_{s_n}\circ \cdots \circ \delta_{s_1}\).

\smallskip

A Mealy automaton \((Q,\Sigma,\delta, \rho)\) is
\emph{reversible\/} if its dual is invertible, that is if
 the functions \(\delta_i\) are permutations of~\(Q\). Note that a
 connected component of a reversible Mealy automaton is always
 strongly connected.

\medskip

An automaton group or semigroup can be seen as acting on a regular
rooted tree representing the language of all words on its alphabet.


\subsection{Growth of a semigroup or of a group}
Let \(H\) be a semigroup generated by a finite set \(S\). The
\emph{length} of an element \(g\) of the 
semigroup, denoted by \(|g|\), is the length of its shortest decomposition: \[|g| = 
\min\{ n\mid \exists s_1,\dots,s_n\in S,\, g=s_1\cdots
s_n\}\:.\]

The \emph{growth function} \(\gamma_H^S\) of the semigroup \(H\) with
respect to the generating set~\(S\) enumerates the elements of~\(H\)
with respect to their length:
\[\gamma_H^S(n) = \#\{g\in H\,;\, |g|\leq n\}\:.\]

\smallskip

The \emph{growth functions} of a group are defined similarly by taking
symmetrical generating sets.

\medskip

The growth functions corresponding to two generating sets are
equivalent~\cite{Mann}, so we may define the \emph{growth} of a group
or a semigroup as the equivalence class of its growth
functions. Hence, for example, a finite (semi)group has a bounded 
growth, while an infinite abelian (semi)group has a polynomial growth,
and a non-abelian free (semi)group has an exponential growth.

\bigskip

It is quite easy to obtain groups of polynomial or exponential
growth. Answering a question of Milnor~\cite{milnor}, Grigorchuk gave
the very first example of an automaton group of intermediate
growth~\cite{grigorch:degrees}: faster than any polynomial, slower than
any exponential, opening thus a new classification criterium for
groups, that has been deeply studied since this seminal
article (see~\cite{GP06} and references therein). This example is an
automaton group. It is now
known as \emph{the Grigorchuk group}.

\subsection{Level-Transitivity}
The action of a (semi)group generated by an invertible Mealy automaton
\(\aut{A}=(Q,\Sigma,\delta,\rho)\) is \emph{level-transitive} if its
restriction to \(\Sigma^n\) has a unique orbit, for any~\(n\) (this
notion is equivalently called \emph{spherically transitive}~\cite{gns}). From a
dual point of view it means that the powers of the dual \(\dual{\aut{A}}\) are
connected, the \emph{\(n\)-th power} of the automaton
\(\dual{\aut{A}}\) being the Mealy automaton
\begin{equation*}
\dual{\aut{A}}^n = \bigl( \ \Sigma^n,Q, (\rho_x\colon \Sigma^n \rightarrow
\Sigma^n)_{x\in Q}, (\delta_{\mot{s}}\colon Q \rightarrow Q
)_{\mot{s}\in \Sigma^n} \ \bigr)\enspace.
\end{equation*}

Note that all the powers of a reversible Mealy automaton are reversible.

The next theorem is proved in~\cite{Kli13}:
\begin{theorem}\label{thm-kli13}
Let \(\aut{A}\) be a reversible automaton with a prime number of
states. If the action of \(\dual{\aut{A}}\) is level-transitive, then
the semigroup \(\presm{\aut{A}}\) is free on the automaton state set.
\end{theorem}

In~\cite{Kli13} the hypothesis of the prime number of states was
erroneously conjectured to be not mandatory. In fact, the Mealy
automaton of Figure~\ref{fig-exLaurent} given by Laurent Bartholdi
(personal communication) does not generate a free semigroup on its
state set, even though its dual generates a level-transitive group.

Although deciding the level transitivity of an automaton group or of
an element of an automaton group are open
problems~\cite[Problems 7.2.1(e+f)]{gns}, this former problem has
received a solution in some cases~\cite{Steinb15} and it is even implemented in
the \GAP packages \FR and \SK~\cite{GAP,FR,SK}.

\bigskip

Note that there exists a previous result linking the
level-transitivity of a group and the freeness of an automaton group
on its sate set. In~\cite{square}, Glasner and Mozes associate to a
Mealy automaton a special graph called \emph{\(VH\)-square
  complex}, based on a natural tiling set. They prove that if the
action of the fundamental group of this graph acts level-transitively,
then the group generated by the Mealy automaton is free.

\subsection{Minimization and Nerode classes}
Let $\aut{A}=(Q,\Sigma,\delta,\rho)$ be a Mealy automaton.

The \emph{Nerode equivalence \(\equiv\) on \(Q\)} is the limit of the
sequence of increasingly finer equivalences~$(\equiv_k)$ recursively
defined by:
\begin{align*}
\forall x,y\in Q,\qquad\qquad x\equiv_0 y & \ \Longleftrightarrow
\ \rho_x=\rho_y\:,\\
\forall k\geqslant 0,\ x\equiv_{k+1} y &
\ \Longleftrightarrow\  \bigl(x\equiv_k y\quad \wedge\quad\forall
i\in\Sigma,\ \delta_i(x)\equiv_k\delta_i(y)\bigr)\:.
\end{align*}

Since the set $Q$ is finite, this sequence is ultimately constant.
For every element~$x$ in~$Q$, we denote by~$[x]$  the
class of~$x$ \wrt the Nerode equivalence, called 
the \emph{Nerode class\/}  of
\(x\). Extending to the \(n\)-th power of \(\aut{A}\), we denote 
by \([\mot{x}]\) the Nerode class in \(Q^n\) of
\(\mot{x}\in Q^n\).

Two states of a Mealy automaton belong to the
same Nerode class if and only if they represent
the same element in the generated semigroup, \ie if and only
if they induce the same action on \(\Sigma^*\).

\medskip

The \emph{minimization} of $\aut{A}$ is the Mealy automaton
\(\mz(\aut{A})=(Q/\mathord{\equiv},\Sigma,\tilde{\delta},\tilde{\rho})\),
where for every $(x,i)$ in $Q\times \Sigma$,
$\tilde{\delta}_i([x])=[\delta_i(x)]$ and
$\tilde{\rho}_{[x]}=\rho_x$.
This definition is consistent with the standard minimization of
``deterministic finite automata'' where instead of
considering the mappings $(\rho_x:\Sigma\to\Sigma)_x$, the computation
is initiated by the separation between terminal and non-terminal
states.

\medskip

The following remarks will be useful for the rest of the paper:

\begin{remark}\label{rem-rel}
If two words of \(Q^*\) are equivalent, so are their images under the
action of any element of~\(\presm{\dual{\aut{A}}}\).
\end{remark}

\begin{remark}\label{rem-samesize}
The Nerode classes of a connected reversible Mealy automaton (\ie a
Mealy automaton with exactly one connected component) have the
same cardinality.
\end{remark}

\begin{remark}\label{rem-finiteness}
It is known from~\cite{klimann_ps:3state} that a reversible automaton
generates a finite semigroup if and only if the sizes of the connected components
of its powers are uniformely bounded. It is straightforward to adapt the proof to
show that a reversible automaton generates a finite semigroup if and
only if the sizes of the minimizations of the connected components of its powers
are uniformely bounded.
\end{remark}

\begin{remark}\label{rem-sameaction}
Let \(\aut{A}\) and \(\aut{B}\) be two reversible connected Mealy automata on
the same alphabet \(\Sigma\), \(x\) some state of \(\aut{A}\), and \(y\) some
state of \(\aut{B}\). If \(x\) and \(y\) have the same action on
\(\Sigma^*\), then \(\mz(\aut{A})\) and \(\mz(\aut{B})\) are
isomorphic; in particular they have the same size. Indeed the image of
\(x\) in \(\aut{A}\) by some word \(\mot{s}\in\Sigma^*\) and the image
of \(y\) in \(\aut{B}\) by this same word \(\mot{s}\) have necessarily
the same action on \(\Sigma^*\), and \(\aut{A}\) and \(\aut{B}\) being strongly
connected (because they are connected and reversible), for every state
of \(\aut{A}\) there is a state of \(\aut{B}\) which acts similarly on
\(\Sigma^*\), and vice-versa.
\end{remark}

\section{Main result}
As already said, Theorem~\ref{thm-kli13} cannot be generalized to any
number of states. Our attempt here is to understand which property,
weaker than freeness, can be deduced from the level-transitivity of 
the group generated by the dual automaton, without any hypothesis on
the size of the state set. To obtain the following result, we need to
add some invertibility hypothesis:

\begin{theorem}\label{thm-main}
Let \(\aut{A}\) be an invertible and reversible Mealy automaton. If the action of
\(\dual{\aut{A}}\) is level-transitive, then the semigroup
\(\presm{\aut{A}}\) has exponential growth.
\end{theorem}

Note that the exponential growth of the semigroup \(\presm{\aut{A}}\)
implies the exponential growth of the group \(\pres{\aut{A}}\).

\bigskip

Let us first look at the structure of the Nerode classes of two consecutive
powers of the state set, in the case of an invertible and reversible Mealy automaton
whose dual generates a level-transitive group.

\begin{lemma}\label{lem-rec1}
Let \(\aut{A}=(Q,\Sigma,\delta,\rho)\) be an invertible and reversible automaton whose
dual generates a level-transitive group.

Let \((C_i)_{1\leq i\leq k}\) be the Nerode classes of \(Q^n\)
  for some \(n\), and \(D\) be a Nerode class of \(Q^{n+1}\). We have
\[D = \bigcup_{q\in Q_D} C_{i_{q,D}} q\quad\text{and}\quad D = \bigcup_{q\in Q'_D} qC_{i'_{q,D}}\:,\]
where \(Q_D\subseteq Q\) and \(Q'_D\subseteq Q\) have the same cardinality, and the
\((i_{q,D})_{q\in Q_D}\) on the one hand and the \((i'_{q,D})_{q\in Q_D}\) on the
other are pairewise distinct.

The automata \(\mz(\aut{A}^n)\) and \(\mz(\aut{A}^{n+1})\) have the
same size if and only if \(Q_D=Q'_D=Q\).
\end{lemma}

\begin{proof}
We prove the first decomposition.

Let \(\mot{u}\in Q^n\) and \(\mot{v}\in [\mot{u}]\), then
\(\mot{u}q\) and \(\mot{v}q\) are in the same
Nerode class of \(Q^{n+1}\). So \(\mot{u}q\in D\) implies
\([\mot{u}]q\subseteq D\).

Let \(C_i\) and \(C_j\) be two different Nerode classes of \(Q^n\), then
it is impossible that \(C_iq\subseteq D\) and \(C_jq\subseteq D\)
because by hypothesis \(q\) induces a bijection on \(\Sigma^*\).

As a consequence, the ratio between the cardinality of a Nerode
class of \(Q^{n+1}\) and the cardinality of a Nerode class of \(Q^n\) (which does
not depend on these classes by Remark~\ref{rem-samesize}) is the
integer \(\#Q_D=\#Q'_D\) which is less than or equal to \(\#Q\). It is
equal to \(\#Q\) if and only if \(\mz(\aut{A}^n)\) and
\(\mz(\aut{A}^{n+1})\) have the same size.
\qed\end{proof}

\begin{proposition}
Let \(\aut{A}\) be an invertible and reversible Mealy automaton whose dual generates a
level-transitive group, and \(n\) be an integer.

If \(\#\mz(\aut{A}^{n+1}) = \#\mz(\aut{A}^n)\), then
\(\#\mz(\aut{A}^{n+2}) = \#\mz(\aut{A}^{n+1})\).
\qed\end{proposition}

\begin{proof}
Let us denote by \(Q\) the state set of \(\aut{A}\), by
\(C_1,\dots,C_k\) the Nerode classes of \(Q^n\), and by 
\(D_1,\dots,D_k\) the Nerode classes of \(Q^{n+1}\) (by hypothesis
\(Q^{n+1}\) has as many Nerode classes as \(Q^n\)).

Let \(r\in Q\): \(rD_1\) is included in some Nerode class of
\(Q^{n+2}\). But by Lemma~\ref{lem-rec1} we have
\[rD_1 = r\bigcup_{q\in Q} C_{i_q}q = \bigcup_{q\in Q} rC_{i_q}q\:,\]
where \(C_{i_q}\) is written for \(C_{i_{q,D_1}}\).
Now, for a fixed \(q\in Q\), \(rC_{i_q}\) is included in some Nerode
class of \(Q^{n+1}\), say \(D_{j_q}\). Hence
\[\bigcup_{q\in Q} D_{j_q}q\] is included in a Nerode class of
\(Q^{n+2}\) and by Lemma~\ref{lem-rec1} we obtain the result.
\qed\end{proof}

We can now prove Theorem~\ref{thm-main}.

\begin{proof}
If \(\dual{\aut{A}}\) is level-transitive, then \(\pres{\dual{\aut{A}}}\) is
infinite and so is \(\presm{\aut{A}}\) (see for example~\cite{AKLMP12,klimann_ps:3state}).
But
if there exists \(n\) such that \(\mz(\aut{A}^n)\) and
\(\mz(\aut{A}^{n+1})\) have the same size, then for any \(m\geq n\),
\(\mz(\aut{A}^m)\) has this same size, and the generated 
semigroup should be finite by Remark~\ref{rem-finiteness}.

So for any \(n\), the ratio of the size of \(\mz(\aut{A}^{n+1})\) and
the size of \(\mz(\aut{A}^{n})\), which is known to be an integer by Lemma~\ref{lem-rec1}, is
at least~2. So the sequence \((\#\mz(\aut{A}^n))_n\) increases and is
componentwise greater than or equal to \((2^n)_n\). By
Remark~\ref{rem-sameaction}, this implies that, for any integer \(n\), the actions induced by
the states of \(\aut{A}^n\) cannot be induced by previous powers of
\(\aut{A}\) and hence \(\presm{\aut{A}}\) has an exponential growth.
\qed\end{proof}

\paragraph{Open problems}
For this result to be fully applicable, it remains of course to find a procedure
to decide if an invertible and reversible Mealy automaton is
level-transitive. Another interesting question would be: can an
invertible and reversible Mealy automaton generate a semigroup of polynomial or
of intermediate growth?

\bibliographystyle{amsplain}
\bibliography{leveltransitive}

\end{document}